\documentclass[a4paper,11pt]{article}
\usepackage[a4paper,margin=2.5cm]{geometry}
\pdfoutput=1
\usepackage{setspace}
\usepackage[authoryear,round]{natbib}

\usepackage[latin1]{inputenc} 
\usepackage[english]{babel}
\usepackage{amsmath,amsthm,amssymb}
\usepackage[T1]{fontenc}
\usepackage{hyperref}
\usepackage{doi}

\newtheorem{thm}{Theorem}
\newtheorem{prop}[thm]{Proposition}

\newtheorem{cor}[thm]{Corollary}

\theoremstyle{definition}
\newcommand{\newdefinition}{\newtheorem}
\newdefinition{df}{Definition}
\newdefinition{rmk}{Remark}

\newcommand{\R}{{\mathbb{R}}}
\newcommand{\eqd}{\overset{\textrm{\,d}}=}

\renewcommand{\vec}{\boldsymbol}
\newcommand{\SIGMA}{\vec\varSigma}
\newcommand{\PI}{\vec\varPi}
\newcommand{\LAMBDA}{\vec\varLambda}
\newcommand{\THETA}{\vec\varTheta}
\newcommand{\GAMMA}{\vec\varGamma}
\newcommand{\rvek}{\mathbf}
\newcommand{\rv}{\mathrm}

\newcommand{\pf}{\vec \xi}
\newcommand{\Unif}{\rvek U}
\newcommand{\Radial}{\rv R}

\newcommand{\T}{{}^{\hspace{-1.25pt}{\top}}}
\newcommand{\conditioningset}{C}

\begin{document}

\title{Weighted selection from elliptical distributions: a stochastic
	representation and an application to portfolio separation}

\author{Nils Chr.~Framstad\thanks{Department of Economics, University of Oslo,
		P.O.\ Box 1095 Blindern, NO-0317 Oslo, Norway.
		E-mail: \href{mailto:ncf+research@econ.uio.no}{\texttt{ncf+research@econ.uio.no}}.}}

\date{\today}

\maketitle

\begin{abstract}
	We represent (weighted-)selection-elliptical distributions as an affine
	combination of the $q$ selection variables plus an elliptical term whose
	direction alone is independent. This form suffices for $q+2$ fund separation
	via first-order stochastic dominance, inter alia relaxing Simaan's (1993) three-fund
	assumptions.
\end{abstract}


\section{Introduction.}
Among the cornerstones of so-called modern portfolio theory \citep{Markowitz1952JF} are portfolio separation (a.k.a.\ \emph{mutual fund\/}) theorems: the property that a large number of investment opportunities can be replaced by a few market indices (``funds'') with no loss of welfare/utility to the investors. \citet{Tobin:LPABTR} initiated this line of results, most canonically with one risky and one riskless fund (so-called \emph{monetary\/} separation). It is beyond the scope of a short communication to give a comprehensive literature overview. This paper falls in line with the direction of \citet{MR58:4217}, who gave a characterization in terms of second-order stochastic dominance for validity across the class of increasing concave von Neumann--Morgenstern utility indices, though also briefly covering first-order stochastic dominance (allowing fewer distributions but all nondecreasing vNM utility indices). Seminal subsequent references include \citet{RePEc:bla:jfinan:v:38:y:1983:i:3:p:745-52} and
\citet{RePEc:eee:jetheo:v:29:y:1983:i:1:p:185-201} for the class of \emph{elliptically distributed\/} risky returns vector, the latter also characterizing the distributions that imply dependence upon mean and variance only. \citet{5d09257c-83da-306d-ade7-ed20d4117e2f} obtained three-fund separation for an
elliptical plus an independent non-elliptical (``skewness'') factor. 

This paper builds upon an earlier work \citep{Framstad20111862} obtaining
$q+2$ fund separation for skew-elliptical distributions  (e.g.\ \cite{MR1867257} and more generally \cite{MR2467905}) conditioned on $q$ variables -- therein, the argument was closely tied to the conditioning structure. We shall see that the conditioning structure is not at all crucial -- what matters is a stochastic form rather than conditioning: an
elliptical whose uniformly distributed direction is independent of everything else, plus
$q$ further random variables which may have arbitrary stochastic dependence with the radial part. We establish such a form beyond skew-elliptical distributions, where we can allow a reweighted draw for selection of the $q$ coordinates (Section~\ref{sec:distributions}). This stochastic representation form alone -- regardless of whether the $q$ additional sources of randomness originate from ellipticity -- is sufficient for $q+2$ fund separation over all non-decreasing vNM indices, using first-order stochastic dominance (Section~\ref{sec:portfolio}).

As a corollary, we get a more general version of the \cite{5d09257c-83da-306d-ade7-ed20d4117e2f} three-fund separation theorem -- his assumptions of both risk aversion and independence can be relaxed. Furthermore we contrast a particular case to 
\citet[Theorem 2]{RePEc:eee:jetheo:v:29:y:1983:i:1:p:185-201}. The discussion in Section~\ref{sec:discussion} will outline how the framework of first-order stochastic dominance is suited for multi-period models.

\section{Generalizations of elliptical distributions.}\label{sec:distributions}
We shall for random variables (boldface if multivariate) use upright capital letters beyond $Q$ % 
in the Latin alphabet; italics letters (Greek and Latin) are non-random. The application to portfolio theory will follow easily once the $n$-variate excess returns distribution can be represented as 
\begin{align}\label{Yrepr}
	\rvek Y = \vec\mu_0\Radial_0 + \vec M\rvek V + \Radial\vec A\Unif_\ell = \vec\mu_0\Radial_0 + \vec\mu_1\rv V_1+\dots+\vec\mu_q\rv V_q+ \Radial\vec A\Unif_\ell 
\end{align}
and we are going to arrive at that for various generalizations of the elliptical class (which, cf.\ \citet{MR629795} is the case $\vec M\rvek V\equiv\vec{0}$, $\Radial_0\equiv 1$). 
In \eqref{Yrepr}, the random elements are $\Radial_0\geq 0$ a.s., $\Radial\geq 0$ a.s.; the $q$-variate $\rvek V$; and finally $\Unif_\ell$ uniformly distributed on the (Euclidean) unit sphere $\mathbb S^{\ell-1}=\{\vec u\in\R^\ell;\; ||\vec u||=1\}$ and assumed \emph{independent of $(\Radial_0,\Radial,\rvek V)$\/} (thus we shall not enter into the dependent-direction generalization of \cite{frahm2004generalized}). Non-random are $\vec\mu_0\in\R^n$, $\vec M\in\R^{n\times q}$, and $\vec A\in\R^{n\times \ell}$. We shall make further assumptions for the application, but mention that $\vec\mu_0\Radial_0$ is there to have an ``unambiguous good'' in portfolio returns. 

We shall see how the skew-elliptical and more general \emph{selection-elliptical\/} families -- and beyond -- fit into the representation \eqref{Yrepr}.

\subsection{Skew-elliptical and (weighted-)selection-elliptical distributions}
The origin of skew-elliptical distribution theory is arguably \citet{azzalini1985class}; he traces the idea back to \citet{Birnbaum1950}. The literature has grown formidable, see e.g.\ the bibliography on his web page \citep{AzzaliniWebpage}.  %\url{http://azzalini.stat.unipd.it/SN/}, 
The selection-elliptical class is much less developed, the phrase being taken from \citet{arellanovalle2006unified}. In their introduction they also mention weighted distributions, which originate with \cite{rao1965discrete}, whose Section~3.4 already allows random weights. Our phrase ``weighted-selection-elliptical'' in Definition~\ref{defn:selection} reflects Rao's terminology, and the fact that we put a weighting $\varrho(\rvek Z)$ only on the \emph{selection\/} variable.

A skew-elliptical distribution is formed by considering some coordinates of an elliptical, conditioned upon the rest belonging to a translated orthant. Like \citet{Framstad20111862} we follow \citet{MR2467905}, skipping the shift parameter (to be reintroduced at the end) and considering
\begin{gather}\label{WZ}
	\begin{pmatrix}\rvek W\\\rvek Z
	\end{pmatrix}=
	\Radial_{\rvek W,\rvek Z}\: \vec L\vec D\Unif_{n+q},\quad \vec D=\begin{pmatrix}\SIGMA^{1/2}&\vec 0_{n\times q}\\ \vec 0_{q\times n}&\LAMBDA^{1/2}
	\end{pmatrix},\quad \vec L=\begin{pmatrix}\vec I_{n\times n}&\vec 0_{n\times q}\\ \PI&\vec I_{q\times q}
	\end{pmatrix}
\end{gather}
starting from an elliptical $\Radial\vec D\Unif$ with block-diagonal $\vec D$ (i.e.\ ``block-uncorrelated'' except no integrability assumed), where the blocks $\SIGMA\in \R^{n\times n}$ and $\LAMBDA\in\R^{q\times q}$ are positive-definite and  $\Unif_{n+q}$ is uniform on the unit sphere $\mathbb S^{n+q-1}$ and independent of the radial $\Radial=\Radial_{\rvek W,\rvek Z}\geq 0$ (not a.s.\ zero); every centered elliptical with nonsingular scatter admits that form on some probability space \citep[Theorem 1]{MR629795}. Then we transform by $\vec L$ to arrive at an $n$-variate $\rvek W$ and a $q$-variate $\rvek Z$.   We note that positive definiteness of $\SIGMA$ and $\LAMBDA$ yields positive definiteness of
\begin{gather}
	\vec L\vec D(\vec L\vec D)\T=\begin{pmatrix}\SIGMA &\SIGMA\PI\T\\
		\PI\SIGMA &\LAMBDA+\PI\SIGMA\PI\T
	\end{pmatrix}.
\end{gather} 

We may want to condition $\rvek W$ upon not only an orthant, but any non-null event for $\rvek Z$ -- indeed we shall allow for a (re-)weighted selection, covering also mixtures of conditionings:
\begin{df}\label{defn:selection}Suppose the probability measure $\mathbb P^{(\varrho)}$ is given by $d\mathbb P^{(\varrho)}=\varrho(\rvek Z)d\mathbb P$ for some Borel $\varrho:\mathbb R^q\to[0,\infty)$ with $\mathbb E[\varrho(\rvek Z)]=1$. The law of 
	$\vec\mu_0+\rvek W$ under $\mathbb P^{(\varrho)}$ is called \emph{weighted-selection-elliptical}; 
	under $\mathbb P^{(\varrho)}$, $\rvek Z$ has law $\varrho\,d\mathbb P_{\rvek Z}$ while conditional 
	laws given $\rvek Z$ are as under $\mathbb P$. The law is called  \emph{selection-elliptical\/} if  $\varrho(\rvek Z)=\frac{1_{\conditioningset}(\rvek Z)}{\mathbb P[\rvek Z\in\conditioningset]}$ for some $\conditioningset$ with $\mathbb P[\rvek Z\in\conditioningset]>0$ (so that $\mathbb P^{(\varrho)}$ is the conditional probability); and, \emph{skew-elliptical\/} if furthermore $\conditioningset$ is a translated orthant.
\end{df} 

These distributions do admit a stochastic representation as well. The decomposition below is the one underlying conditional distributions of elliptical vectors \citep[proof of Corollary~5]{MR629795}, there in law; we want it \emph{statewise.} Define 
\begin{gather}\label{vekM}
	\vec M = \SIGMA\PI\T(\LAMBDA+\PI\SIGMA\PI\T)^{-1}\\
	\label{vekA}\vec A=\bigl(\SIGMA-\vec M(\LAMBDA+\PI\SIGMA\PI\T)\vec M\T\bigr)^{1/2}=\bigl(\SIGMA^{-1}+\PI\T\LAMBDA^{-1}\PI\bigr)^{-1/2}\\
	\label{vekGammaTheta} \GAMMA=\begin{pmatrix}\vec A&\vec0\\\vec0&(\LAMBDA+\PI\SIGMA\PI\T)^{1/2}\end{pmatrix}\quad\text{and}\quad\THETA=\GAMMA^{-1}\begin{pmatrix}\vec I&-\vec M \\ \vec 0&\vec I
		\end{pmatrix}\vec L\vec D
\end{gather}Then $\vec M$ represents the prediction and $\vec A$ the positive-definite scale of the residual, so that the scale $\GAMMA$ of the elliptical $\left(\begin{smallmatrix}\rvek W-\vec M\rvek Z\\\rvek Z
\end{smallmatrix}\right)$ is block-diagonal: 
\begin{gather}\label{W-MZandZ}
	\begin{pmatrix}\rvek W-\vec M\rvek Z\\\rvek Z
	\end{pmatrix}=\Radial_{\rvek W,\rvek Z}\;\GAMMA\GAMMA^{-1}\begin{pmatrix}\vec I&-\vec M \\ \vec 0&\vec I
	\end{pmatrix}\vec L\vec D\Unif_{n+q}=\Radial_{\rvek W,\rvek Z}\;\GAMMA\THETA\Unif_{n+q}
\end{gather}
Furthermore, $\THETA$ is orthogonal, as $\left(\begin{smallmatrix}\vec I&-\vec M \\ \vec 0&\vec I
\end{smallmatrix}\right)\vec L\vec D(\vec L\vec D)\T\left(\begin{smallmatrix}\vec I&\vec0 \\ -\vec M\T&\vec I
\end{smallmatrix}\right)=\GAMMA\GAMMA\T$; thus 
 $\breve\Unif:=\THETA\,\Unif_{n+q}$ is also uniform on $\mathbb S^{n+q-1}$. Decompose $\breve\Unif\T=(\breve\Unif^*{}\T,\breve\Unif_*\T)$ for top $n$ and bottom $q$ components: 
\begin{gather}\label{eq:W-MZ}
	\rvek W-\vec M\rvek Z = \Radial_{\rvek W}\,\vec A\widetilde{\Unif}_n,\qquad\text{with}\quad\Radial_{\rvek W}:=\Radial_{\rvek W,\rvek Z}||\breve\Unif^*||,\qquad \widetilde{\Unif}_n:=\frac{\breve\Unif^*}{	||\breve\Unif^*||}
\end{gather}
holds statewise (a.s.), permitted as $||\breve\Unif^*||\neq 0$ a.s.; for the following proposition, we need to establish the distributional properties under $\mathbb P^{(\varrho)}$: 
\begin{prop}\label{prop:stochrep}
	Suppose $\SIGMA$ and $\LAMBDA$ positive definite and let $\mathbb P^{(\varrho)}$ be as in Definition~\ref{defn:selection}. Let $\vec\mu_0\in \R^n$ be non-random and  $\vec M$ and $\vec A$ be given by \eqref{vekM}--\eqref{vekA}. Put $\ell=n$, $\Radial_0\equiv 1$, $\rvek V=\rvek Z$, and 
	\begin{align}\label{eq:R}
		\Radial=\Radial_{\rvek W}=\bigl(\Radial_{\rvek W,\rvek Z}^2-\rvek V\T(\LAMBDA+\PI\SIGMA\PI\T)^{-1}\rvek V\bigr)^{1/2}.
	\end{align}
	Then $\vec\mu_0+\rvek W=\vec\mu_0+\vec M\rvek V+\Radial\vec A\widetilde\Unif_n$ holds 
	$\mathbb P$-a.s., hence $\mathbb P^{(\varrho)}$-a.s. Furthermore, under $\mathbb P^{(\varrho)}$, $\rvek V$ has law $\varrho\, d\mathbb P_{\rvek Z}$, and $\widetilde\Unif_n$ is uniform on $\mathbb S^{n-1}$ and independent 
	of $(\rvek V,\Radial)$. 
	
	In particular, every weighted-selection-elliptical law admits a stochastic representation \eqref{Yrepr}.
\end{prop}
\begin{proof}
	The statewise identity is \eqref{eq:W-MZ}. \eqref{eq:R} follows from $\rvek Z\T(\LAMBDA+\PI\SIGMA\PI\T)^{-1}\rvek Z=\Radial_{\rvek W,\rvek Z}^2||\breve\Unif_*||^2=\Radial_{\rvek W,\rvek Z}^2-\Radial_{\rvek W,\rvek Z}^2||\breve\Unif^*||^2=\Radial_{\rvek W,\rvek Z}^2-\Radial_{\rvek W}^2$. What remains is to show that $\widetilde\Unif_n$ is independent of $(\rvek Z,\Radial_{\rvek W})$ and uniform on $\mathbb S^{n-1}$ under $\mathbb P^{(\varrho)}$. This will follow once established under $\mathbb P$, for then 
	\begin{align}\label{eq:ErhoU}
	\mathbb E^{(\varrho)}\bigl[g(\widetilde\Unif_n)h(\rvek Z,\Radial_{\rvek W})\bigr]
	=\mathbb E\bigl[g(\widetilde\Unif_n)\,h(\rvek Z,\Radial_{\rvek W})\varrho(\rvek Z)\bigr]
	=\mathbb E[g(\widetilde\Unif_n)]\;\mathbb E^{(\varrho)}\bigl[h(\rvek Z,\Radial_{\rvek W})\bigr]
	\end{align}(any bounded measurable $g$, $h$). Case $h\equiv1$ yields $\mathbb E^{(\varrho)}\bigl[g(\widetilde\Unif_n)\bigr]
	=\mathbb E[g(\widetilde\Unif_n)]$ thus same distribution for every $\varrho$; plug that for $\mathbb E[g(\widetilde\Unif_n)]$ in \eqref{eq:ErhoU} and we have independence.
	
	The argument under $\mathbb P$: $\breve\Unif$ inherits the independence of 
	$\Radial_{\rvek W,\rvek Z}$ from $\Unif_{n+q}$, as $\THETA$ is non-random. From 
	\cite[Lemma 2]{MR629795}, $\|\breve\Unif^*\|$, $\widetilde\Unif_n$ and 
	$\breve\Unif_*/\|\breve\Unif_*\|$ are independent, with $\widetilde\Unif_n$ uniform on 
	$\mathbb S^{n-1}$. Hence $\Radial_{\rvek W,\rvek Z}$, $\widetilde\Unif_n$, 
	$\|\breve\Unif^*\|$, $\breve\Unif_*/\|\breve\Unif_*\|$ are mutually independent. $\Radial_{\rvek W}$ and $\rvek Z=\Radial_{\rvek W,\rvek Z}(\LAMBDA+\PI\SIGMA\PI\T)^{1/2}
	\breve\Unif_*$ are functions of the first and third of these, 
	thus $\widetilde\Unif_n$ is independent of $(\rvek Z,\Radial_{\rvek W})$.
\end{proof}

\begin{rmk}\label{remark}
To our knowledge, this representation has not been applied in the weighted-selection case, nor in the selection-elliptical case except when the underlying is Gaussian or we condition upon a translated orthant. For the skew-elliptical situation we refer to e.g.\ \citet[formula (8)]{arellanovalle2010multivariate}, while for the selection-\emph{normal\/} (i.e.\ underlying Gaussian) it appears in \cite{arellanovalle2006unified}, who discuss the difficulty of the dependence structure; only in the case of an underlying Gaussian is the radial independent of $\rvek V$ under $\mathbb P$ \citep[Theorem 6 and subsequent remark (d)]{MR629795}. However, the uniform-on-unit-sphere is still independent no matter the distribution of the radial. 
\end{rmk}

\section{Univariate projections and the portfolio selection application.}\label{sec:portfolio} This section shall establish portfolio separation -- the property that a full market can be represented by a set of indices (``funds'') at no welfare loss -- when the risky excess returns obey \eqref{Yrepr}. Although it is customary to assume risk aversion, we shall make the weaker assumption of preference orderings consistent with first-order stochastic dominance, i.e.\ valid for all von Neumann--Morgenstern preferences with nondecreasing utility index.

Consider an agent initially owning $w_0$ and making a single-period portfolio choice over a ``risk free'' opportunity that yields $\rv Y_0$ and $n$ ``risky'' opportunities that yield $\rv Y_0\vec{1}+\rvek Y$, so that $\rvek Y$ is the vector of \emph{excess\/} returns. Investing amounts $\xi_j$ in risky opportunity number $j$ and the remaining holding $w_0-\pf\T\vec{1}$ in the riskless, they are left with terminal wealth 
\begin{align}\label{terminalwealth}
	w_0\rv Y_0 + \pf\T\rvek Y
\end{align}
Common is to assume $\rv Y_0$ non-random (cf.\ the typical terminology ``risk free'' or ``riskless''), but we can just specify $\rvek Y$ conditionally on $\rv Y_0$. We can also accommodate the classical situation where no riskless opportunity is available, by introducing it fictitiously and constraining $\pf$ to the hyperplane $\pf\T\vec{1}=w_0$; in that case, just put $\rv Y_0$ to zero. 

Although it would be natural to impose a no-arbitrage assumption, we only need that every non-trivial univariate projection has either some exposure to $\rvek V$ or some exposure to the elliptical:
\begin{align}\label{eq:nonsingular}
	\vec A\vec A\T+\vec M\vec M\T\qquad\text{is positive definite.}
\end{align}
Then  $q+2$ fund separation (where the last fund is the risk-free opportunity, except this is replaced by a different one if $w_0-\pf\T\vec{1}\equiv 0$ is imposed) can be formulated as follows:

\begin{prop}[$(2+\text{\rm rank}(\vec M))$ fund separation]\label{prop:main}
	Let $q<n$. Consider formula \eqref{terminalwealth} where, conditional on $\rv Y_0$, $\rvek Y$ is distributed as in \eqref{Yrepr} where $\Radial_0\geq 0$ a.s.\ (and w.l.o.g.\  $\Radial\geq 0$ a.s.), and where $\vec A$, $\vec M$, $\vec \mu_0$ are non-random and \eqref{eq:nonsingular} holds. Let 
	\begin{gather}
		\vec f_j=(\vec A\vec A\T+\vec M\vec M\T)^{-1} \vec \mu_j,\qquad \text{for } j=1,\dots,q
	\end{gather}
	and also for $j=0$ in the case  $\vec\mu_0\not\in\text{\rm col}(\vec M)$. If  $\vec\mu_0$ lies in the column space, let instead  $\vec f_0$ be a(ny) nonzero vector $\in  \text{\rm ker}(\vec M\T)$.
		
	Then for each $\pf\in\R^n$ there exist real $c_0,\dots,c_q$ and $\varpi\geq0$, such that with $\pf^*=c_0\vec f_0+\dots+c_q\vec f_q$, we have the first-order stochastic dominance
	\begin{align}\label{eq:FOSD}
		w_0\rv Y_0 + \pf^*{}\T\rvek Y\eqd w_0\rv Y_0 + \pf\T\rvek Y + \varpi \Radial_0,\qquad \varpi\Radial_0\geq 0\ \text{a.s.}
	\end{align}
	
	Furthermore, to any linear constraint $\pf\T\vec a=\alpha$ (or $\leq \alpha$) there will be (unless already spanned) another $\vec f_{\vec a}=(\vec A\vec A\T+\vec M\vec M\T)^{-1} \vec a$; i.e., as if expanding $\vec M$ by another column $\vec a$. In particular, that goes for the constraint $\pf\T\vec{1}=w_0$ which removes the risk-free opportunity (thus also as a fund) and also results in $q+2$ fund separation. 
\end{prop}
Before passing to the proof, we note that the typical textbook exposition valid for risk-averse agents will have fewer  degeneracies and will not need the ``alternative case dispersion-only'' $\vec f_0$. 
\begin{proof}Assume w.l.o.g.\ that $\vec M$ has full rank; if not, rewriting $\vec M\rvek V$ does not change kernel nor column space (nor the \emph{span\/} of the funds). The distribution of the portfolio return depends on $\pf\T\vec\mu_0$ (a good, as $\Radial_0\geq 0$), $\vec d:=\vec M\T\pf$ and the value $\sigma^2:=\pf\T\vec A\vec A\T\pf$ of the quadratic form. For a $\pf^*$ so that $\vec M\T\pf^*=\vec d$,  $\pf^*{}\T\vec A\vec A\T\pf^*=\sigma^2$, and $\varpi:=(\pf^*-\pf)\T\vec\mu_0$ is $\geq 0$, the characteristic functions of either side of \eqref{eq:FOSD} match. To see this, condition first upon everything but $\Unif_\ell$, which has chf  $=\psi_\ell(\vec\theta\T\vec\theta)$. By independence, integrate out $\Unif_\ell$ to get the equality:
	\begin{align}
\begin{split}
			\mathbb E\bigl[\exp\bigl(\mathrm i\vartheta (w_0\rv Y_0+\pf\T\rvek Y + \varpi\Radial_0)\bigr)\bigr]&=
	\mathbb E\bigl[\exp\bigl(\mathrm i\vartheta (w_0\rv Y_0+\pf^*{}\T\vec\mu_0\Radial_0+\vec d\T\rvek V)\bigr)\cdot\psi_\ell(\vartheta^2\Radial^2\sigma^2)\bigr]\\
	&=\mathbb E\bigl[\exp\bigl(\mathrm i\vartheta (w_0\rv Y_0+\pf^*{}\T\rvek Y)\bigr)\bigr]
\end{split}
	\end{align}
	
	In the presence of  $\vec M\T\pf=\vec d$, then $\pf\T\vec A\vec A\T\pf=\pf\T(\vec A\vec A\T+\vec M\vec M\T)\pf-\vec d\T\vec d$, so we can maximize $\pf\T\vec\mu_0$ subject to  $\vec M\T\pf=\vec d$ and $\pf\T(\vec A\vec A\T+\vec M\vec M\T)\pf=\sigma^2+\vec d\T\vec d$ (forming an ellipsoid, so the max exists). The Fritz John conditions:
	\begin{align}
		\vec\mu_0\lambda_0 + \vec M\vec\lambda = \lambda^*\cdot(\vec A\vec A\T+\vec M\vec M\T)\pf^*
	\end{align}
	with multipliers not all zero. Provided $\lambda^*\neq 0$, we must have 
	\begin{align}\label{eq:pfSTAR}
		\pf^*=\tfrac 1{\lambda^*}\bigl[\lambda_0(\vec A\vec A\T+\vec M\vec M\T)^{-1}\vec\mu_0 +\dots+\lambda_q(\vec A\vec A\T+\vec M\vec M\T)^{-1}\vec\mu_q\bigr]
	\end{align}
	Fritz John also covers ($\lambda_0=0$) the point $\pf_*$ which subject to $\vec M\T\pf=\vec d$ minimizes $\pf\T\vec A\vec A\T\pf$ (equivalently, minimizes $\pf\T(\vec A\vec A\T+\vec M\vec M\T)\pf$). The case $\lambda^*=0$ entails $\vec\mu_0\in\text{\rm span}\{\vec \mu_1,\dots,\vec\mu_q\}$ (by full rank of $\vec M$), so assume now so. Let then instead $\vec f_0\neq\vec0$ be orthogonal to the span, consider $\pf^*=\pf_*+\eta \vec f_0$ and vary $\eta$. Then $\pf^*{}\T[\vec \mu_0\Radial_0+\vec M\rvek V]$ is independent of $\eta$ (with $\vec \mu_0$ in the span,  $\vec M\T\pf$ pins down $\pf\T\vec\mu_0$), while the quadratic form ranges from its minimum at $\eta=0$ and to infinity, since $\vec A\vec A\T$ is positive definite on $\text{\rm ker}(\vec M\T)$; choose $\eta$ so that $\pf^*{}\T\vec A\vec A\T\pf^*=\pf\T\vec A\vec A\T\pf$, and we have identical distribution (i.e.\ $\varpi=0$).
	
	Any additional linear constraints on the portfolio itself can be accommodated straightforward through the same Lagrange/Fritz John argument, including the absence of risk-free opportunity.
\end{proof}
To recover -- and generalize -- the $2+\text{\rm rank}(\PI)$ funds of \citet{Framstad20111862}, the following is arguably more a corollary to the proof than to the proposition statement:
\begin{cor}\label{cor:recoverold}
	For the (weighted-)selection-elliptical situation of Proposition~\ref{prop:stochrep}, funds $\vec f_1,\dots,\vec f_q$ can be taken as the columns of $\PI\T$ and in the situation where $\vec\mu_0$ is not in the column space of $\SIGMA\PI\T$, then $\vec f_0$ can be taken to be $\SIGMA^{-1}\vec\mu_0$. In the degenerate case, take $\vec 0\neq \vec f_0\in\textrm{\rm ker}(\PI\SIGMA)$. 
\end{cor}
\begin{proof}
From \eqref{vekM}, $\vec M\T\pf= (\LAMBDA+\PI\SIGMA\PI\T)^{-1}\PI\SIGMA\pf$ equals $\vec d$ iff $\PI\SIGMA\pf=\vec\delta:=(\LAMBDA+\PI\SIGMA\PI\T)\vec d$. Furthermore, same constraint and \eqref{vekA} yields $\pf\T\vec A\vec A\T\pf=\pf\T\SIGMA\pf -\pf\T\vec M(\LAMBDA+\PI\SIGMA\PI\T)\vec M\T\pf$ to equal $\pf\T\SIGMA\pf-d_0^2$, $d_0^2$ being a constant. Thus for the non-degenerate part we can instead maximize $\pf\T\vec\mu_0$ subject to $\pf\T\SIGMA\pf=\sigma^2+d_0^2$ and $\PI\SIGMA\pf=\vec\delta$, the values taken from the candidate portfolio to be dominated. Repeat the first order condition argument as $\vec\mu_0\hat\lambda_0 + \SIGMA\PI\T\hat{\vec\lambda} = \hat\lambda^*\SIGMA\pf^*$ and pre-multiply by $\SIGMA^{-1}$ which is assumed to exist. If $\PI$ does not have full rank, then remove one or more equation from the constraint $\PI\SIGMA\pf=\vec\delta$. For the degenerate case, minimizing either quadratic form s.t.\ the linear constraint, will yield $\pf_*$.
\end{proof}

\subsection{Case $q=1$, the distributions of \cite{RePEc:eee:jetheo:v:29:y:1983:i:1:p:185-201} and \cite{5d09257c-83da-306d-ade7-ed20d4117e2f}, and mean--dispersion utility}
	The Simaan three-fund separation follows as a special case of $q=1$ -- Simaan does assume independence between our $\rvek V$ and our $\Radial$, as well as $\Radial_0\equiv 1$ and $\rv Y_0$ non-random. As we saw, we do not need $\Radial$ to be independent of everything else, only $\Unif$. We still get two funds from the generalization of the elliptical $\Radial_0 \vec\mu_0 + \Radial\vec A\Unif_\ell$, and then one more for each additional noise factor. Also Simaan's assumption of concave utility index is superfluous.
	
	The seminal results of \citet{RePEc:eee:jetheo:v:29:y:1983:i:1:p:185-201}
	find the elliptical class for his Theorem 1. For the main branch of his Theorem 2 it is our case $\vec\mu_0=0$, $q=1$, $\vec\mu_1\not\in\text{\rm col}(\vec A)$, $\ell=n-1$, square integrability and $\mathbb E\rv V\neq 0$ (his non-elliptical noise can attain either sign, so we cannot put it on $\vec \mu_0$). Then, Chamberlain characterizes the distributions over all portfolios when there is no riskless opportunity. He does however obtain \emph{two\/} fund separation by free disposal -- our method needs a third. 
	
	Nevertheless, assuming $\rv V$ has nonzero mean, then for our distribution class with $q=1$, $\vec\mu_0=0$ and non-random $\rv Y_0$, each univariate projection \emph{is indeed\/} determined by its (finite) mean and (finite) variance, though it is not within Chamberlain's Theorem 1. The reason lies in what that theorem actually characterizes: all possible terminal wealths, when also initial wealth is allowed to vary, so an inferior portfolio can be compensated for with higher initial wealth. In a different discourse of comparison, one may rather want to take initial wealth as a parameter and order the $\{\pf\T\rvek Y\}_{\pf\in\R^n}$, in which case also this larger distribution class works. This is of course in plain sight in Chamberlain, but might have been overlooked in the misunderstandings (criticized e.g.\ by \cite{markowitz2012great}) of conditions justifying mean--variance analysis.

\section{Discussion.}\label{sec:discussion}
We have seen that the stochastic representation form shared by skew-ellipticals and generalizations, is in itself sufficient for the portfolio separation result: compared to the elliptical case that requires two funds, we get one more per additional noise source. Thus that result comes from the properties of additive noise and \emph{conditional\/} ellipticity, that these distributions share, but not from the specific skew-elliptical or (weighted-)selection-elliptical construction. Inspecting the structure, we see that the elliptical noise source need not be independent, as long as its directional component is.

Unlike most of the portfolio optimization literature, we employ first-order stochastic dominance. We mention that this is the framework to use in multiperiod models: There is no straightforward dynamic concept of second-order stochastic dominance, but first-order dominance  \emph{does\/} generalize to Polish spaces with a partial order \citep{MR0494447}, or even merely a pre-order \citep{lindvall1999strassen}. Let us sketch a setup: If period to period wealth evolves according to \eqref{terminalwealth} with the random variables drawn anew independent of the past, then an agent applying $\pf^*$ can then immediately after the realization consume an extra $(\pf^*-\pf)\T\vec\mu_0\Radial_0$ while keeping post-consumption wealth process at same law (run the original strategy on a rotation of $\Unif$). This argument of course requires the model to have intermediate consumption (or allow intermediate free disposal) and also $\Radial_0$ be separately observable after realization -- in a typical model it is $\equiv 1$. See \citet{doi:10.1080/17442508.2015.1132218} for a continuous-time treatment.

Apart from that, the proof herein also shows that the funds sufficient for risk-averse agents in our model, also will be demanded by all vNM agents of increasing but not concave utility index -- and only in degenerate cases those agents will require one more, and then only if they are seeking dispersion; that also covers agents who take into account the probability of exceeding some wealth target.

\bibliographystyle{abbrvnat}
\bibliography{bib2026clean}

\end{document}